%% file: Na_Co_DQSA.tex
\documentclass[journal,10pt,twocolumn]{IEEEtran}
\usepackage{epsfig}
\usepackage{subfigure}
\usepackage{amssymb}
\usepackage{amsbsy}
\usepackage{amsmath}
\usepackage{cite}
\usepackage{url}

\usepackage{color}
\usepackage{float}
\usepackage{times,amsmath,epsfig}
\usepackage{xspace,latexsym,syntonly}
\usepackage{amssymb}
\usepackage{amsfonts}
\usepackage{textcomp}
\usepackage{subfigure}
\usepackage{amsbsy}
\usepackage{cite}
\input{macros}

\begin{document}
\title{Deep Multi-User Reinforcement Learning for Distributed Dynamic Spectrum Access}
\author{Oshri Naparstek and Kobi Cohen
\thanks{Oshri Naparstek is with the Rafael Advanced Defense Systems Ltd., Haifa 31021, Israel. Email: oshrin@rafael.co.il}
\thanks{Kobi Cohen is with the Electrical and Computer Engineering Department, Ben-Gurion University of the Negev, Beer Sheva 8410501 Israel. Email:yakovsec@bgu.ac.il}
\thanks{This work has been submitted to the IEEE for possible publication. Copyright may be transferred without notice, after which this version may no longer be accessible.}
\thanks{A short version of this paper that introduces the algorithm and preliminary experimental results was presented at the IEEE Global Communications Conference (GLOBECOM), 2017 \cite{naparstek2017deep}.}
}
\date{}
\maketitle

\begin{abstract}
\label{sec:abstract}
We consider the problem of dynamic spectrum access for network utility maximization in multichannel wireless networks. The shared bandwidth is divided into $K$ orthogonal channels. In the beginning of each time slot, each user selects a channel and transmits a packet with a certain transmission probability. After each time slot, each user that has transmitted a packet receives a local observation indicating whether its packet was successfully delivered or not (i.e., ACK signal). The objective is a multi-user strategy for accessing the spectrum that maximizes a certain network utility in a distributed manner without online coordination or message exchanges between users. Obtaining an optimal solution for the spectrum access problem is computationally expensive in general due to the large state space and partial observability of the states. To tackle this problem, we develop a novel distributed dynamic spectrum access algorithm based on deep multi-user reinforcement leaning. Specifically, at each time slot, each user maps its current state to spectrum access actions based on a trained deep-Q network used to maximize the objective function. Game theoretic analysis of the system dynamics is developed for establishing design principles for the implementation of the algorithm. Experimental results demonstrate strong performance of the algorithm.
\end{abstract}

\def\keywords{\vspace{.5em}
{\bfseries\textit{Index Terms}---\,\relax%
}}
\def\endkeywords{\par}
\keywords
Wireless networks, dynamic spectrum access, medium access control (MAC) protocols, multi-agent learning, deep reinforcement learning.

\section{Introduction}
\label{sec:intro}

The increasing demand for wireless communication, along with spectrum scarcity, have triggered the development of efficient dynamic spectrum access (DSA) schemes for emerging wireless network technologies. A good overview of various DSA models for medium access control (MAC) design can be found in \cite{zhao2007survey}. In this paper we mainly focus on DSA in the open sharing model among users that acts as the basis for enabling a large number of users to access and share the same limited frequency band. We consider a wireless network with $N$ users sharing $K$ orthogonal channels (e.g., OFDMA). In the beginning of each time slot, each user selects a channel and transmits its data with a certain transmission probability (i.e., Aloha-type narrowband transmission). After each time slot, each user that has transmitted a packet receives a local binary observation indicating whether its packet was successfully delivered or not (i.e., ACK signal). The goal of the users is to maximize a certain network utility in a distributed manner without online coordination or exchanging messages.

\subsection{Learning Algorithms for Dynamic Spectrum Access}
\label{ssec:intro_learning}

Developing distributed optimization and learning algorithms for managing efficient spectrum access among users have attracted much attention in past and recent years (see Section \ref{ssec:related} for a detailed discussion on related work). Complete information about the network state is typically not available online for the users, which makes the computation of optimal policies intractable in general \cite{zhao2007decentralized}. While optimal structured solutions have been developed for some special cases (e.g., \cite{Ahmad_2009_Optimality, Wang_2012_Optimality, cohen2014restless} and references therein), most of the existing studies have focused on designing spectrum access protocols for specific models so that efficient (though not optimal) and structured solutions can be obtained. However, model-dependent solutions cannot effectively adapt in general for handling more complex real-world models. Model-free Q-learning was used in \cite{li2010multiagent} for Aloha-based protocol in cognitive radio networks. Handling large state space and partial observability, however, becomes inefficient under Q-learning (see Section \ref{sec:network} for details on Q-learning).

\subsection{Deep Multi-User Reinforcement Learning for Dynamic Spectrum Access}
\label{ssec:intro_deep}

\emph{Our goal is to develop a distributed learning algorithm for dynamic spectrum access that can effectively adapt for general complex real-world settings, while overcoming the expensive computational requirements due to the large state space and partial observability of the problem. We adopt a deep multi-user reinforcement learning approach to achieve this goal.}

Deep reinforcement learning (DRL) (or deep Q-learning) has attracted much attention in recent years due to its capability to provide a good approximation of the objective value (referred to as Q-value) while dealing with very large state and action spaces. In contrast to Q-learning methods that perform well for small-size models but perform poorly for large-scale models, DRL combines a deep neural network with Q-learning, referred to as Deep Q-Network (DQN), for overcoming this issue. The DQN is used to map from states to actions in large-scale models so as to maximize the Q-value (for more details on DRL and related work see Sections \ref{ssec:related} and \ref{sec:network}). In DeepMind's recently published Nature paper \cite{mnih2015human}, a DRL algorithm was developed to teach computers how to play Atari games directly from the on-screen pixels, and strong performance was demonstrated in many tested games. In \cite{foerster2016learning2}, the authors developed DRL algorithms for teaching multiple players how to communicate so as to maximize a shared utility. Strong performance was demonstrated for several players in MNIST games and the switch riddle. In recent years, there is a growing attention on using DRL methods for other various fields. Other recent studies can be found in \cite{li2017deep, puzanov2018deep} and references therein.

Due to the large state space and the partially observed nature of spectral management among wireless connected devices, we postulate that incorporating DRL methods in the design of DSA algorithms has a great potential for providing effective solutions to real-world complex spectrum access settings, which motivates the research in this paper.

\subsection{Contributions}
\label{ssec:main_results}

We focus on developing a DSA algorithm based on Aloha-type random access protocol. Aloha-based protocols are popular tools primarily because of their ease of implementation and their random access. Simple transmitters can randomly access a channel without spectrum sensing or centralized controller, as opposed to CSMA-type or central-assisted schemes. Furthermore, Aloha-based protocols are much simpler to implement in a hidden terminals environment. Finally, for low loads, Aloha-based protocols may be preferred due to their low delay.

Using DRL methods in the design of spectrum access protocols is a new research direction, motivated by recent developments of DRL in various other fields, and very little has been done in this direction so far. The proposed approach is fundamentally different from existing DRL-based methods for DSA \cite{wangdeep, wang2018deep, challita2017proactive} in the following aspects: it handles a different environment dynamics; it optimizes performance with respect to a more general network utility; and a new DQN architecture is developed with lower complexity implementations (for more details on existing DRL-based methods for DSA see Section \ref{sec:existing}). We believe that the methods developed in this paper can serve as the basis for developing distributed learning algorithms to other resource management problems as well. The contribution of this paper is threefold:

\paragraph{Algorithm development for multi-user DSA with low complexity} We develop a novel deep multi-user reinforcement learning-based algorithm that allows each user to adaptively adjust its transmission parameters with the goal of maximizing a certain network utility. The algorithm can effectively adapt to topology changes, different objectives, and different finite time-horizons. The algorithm is executed without continuing online coordination or message exchanges between users. Furthermore, spectrum sensing or central control are not used in the algorithm. While offline, we train the multi-user DQN at a central unit to maximize the objective function (in contrast to \cite{challita2017proactive}, where the DQN was trained at each base-station). Since the network state is partially observable for each user, and the dynamics is non-Markovian and determined by the multi-user actions (in contrast to \cite{wangdeep, wang2018deep} that handle the single-user case), we use Long Short Term Memory (LSTM) layer that maintains an internal state and aggregate observations over time. This gives the network the ability to estimate the true state using the past partial observations. Furthermore, we incorporate the dueling DQN method used to improve the estimated Q-value due to the occurrence of bad states regardless of the taken action \cite{wang2015dueling}. Since the experience replay method \cite{mnih2013playing}, \cite{mnih2015human}, used in \cite{wangdeep, wang2018deep} for single-user DSA, is undesirable when handling a multi-user learning for DSA due to interactions among users, we collect $M$ episodes at each iteration and create target values for all the episodes.

After completion of the training phase, the users only need to update their DQN weights by communicating with the central unit. In real-time, at each time slot, each user maps its local observation to spectrum access actions based on the trained DQN.

The proposed algorithm is very simple for implementation using simple software defined radios (SDRs). The expensive computations at the training phase are done offline by a centralized powerful unit (e.g., cloud, or network edge), while updating the DQN is rarely required (e.g., once per weeks, months, only when the environment characteristics have been significantly changed and no longer reflects the training experiences). An illustration is provided in Figure \ref{fig:MEC_CR}.
\begin{figure}[htbp]
\centering \epsfig{file=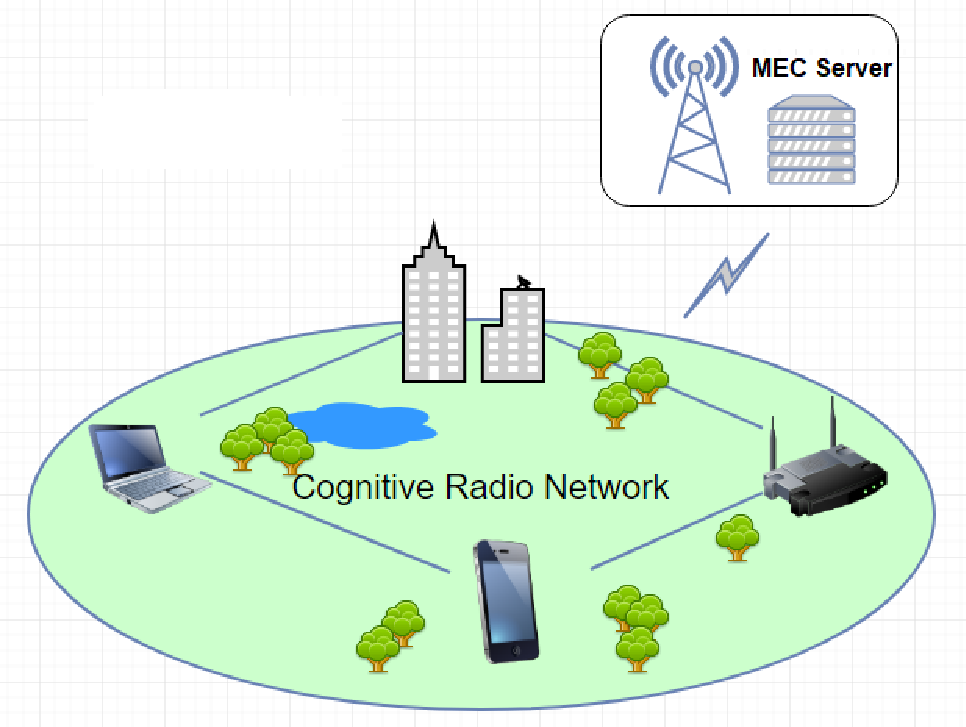,
width=0.35\textwidth}
\caption{An illustration of the network architecture. The expensive computations at the training phase are done offline by the MEC server, located with the wireless access point (e.g., base station). The SDRs update the DQN from time to time (only when the environment characteristics have been significantly changed and no longer reflects the training experiences).}
\label{fig:MEC_CR}
\end{figure}

\paragraph{Analyzing the multi-user dynamics for establishing fundamental algorithm design principles} We use game theoretic analysis in the development of the algorithm that provides us useful tools for modeling and analyzing the multi-user dynamics in Section \ref{sec:analysis}. For a non-cooperative utility, we show that distributed training leads to inefficient subgame perfect equilibria. Thus, we develop a mechanism that restricts the strategy space for all users when training the DQN, referred to as common training, so that it avoids convergence to those inefficient operating points. For a cooperative utility, we develop the first DRL-based approach for DSA that directly optimizes a global system-wide fairness utility. Since the reward for each user is no longer aggregated over time and depends on the common global utility, users receive a common global reward only at the end of the episode. However, it is well known that receiving delayed rewards decreases the training efficiency. Hence, for handling this challenge, we exploit the inherent structure of the objective function to design a reward which is aggregated over time and approximates well the system-wide global utility.

\paragraph{Experimental study} We present extensive numerical experiments for demonstrating the capability of the proposed algorithm to effectively adapt to different problem settings. Under both cooperative and non-cooperative network utilities, we observed that users effectively learn in a fully distributed manner only from their ACK signals how to access the channels so as to increase the channel throughput by reducing the number of idle time slots and collisions. Specifically, the proposed algorithm achieves up to twice the channel throughput as compared to slotted-Aloha with optimal transmission probabilities.

\section{Existing DRL-Based Methods for DSA and Other Related Work}

\subsection{Existing DRL-Based Methods for DSA}
\label{sec:existing}

Developing DRL-based methods for solving DSA problems is a new research direction, motivated by recent developments of DRL in various other fields, and few works have been done in this direction recently. We discuss next the very recent studies on this topic which are relevant to the problem considered in this paper. In \cite{wangdeep, wang2018deep}, the authors developed a spectrum sensing policy based on DRL for a single user who interacts with an external environment. The multi-user setting considered here, however, is fundamentally different in environment dynamics, network utility, and algorithm design.
In \cite{challita2017proactive}, the authors studied a non-cooperative spectrum access problem in a different setting, in which multiple agents (i.e., base-stations in their model) compete for channels and aim at predicting the future system state using LSTM layer with REINFORCE algorithm. The neural network was trained at each agent. The problem formulation in \cite{challita2017proactive} is non-cooperative in the sense that each agent aims at maximizing its own utility, while using the predicted state to reach a certain fair equilibrium point. Our algorithm and problem setting are fundamentally different. First, our algorithm uses LSTM with DQN which is different from the algorithm in \cite{challita2017proactive}. Second, in our algorithm, the DQN is trained for all users at a single unit (e.g., cloud), which is more suitable to various cognitive radio networks and Internet of Things (IoT)-based applications, in which cheap SDRs only need to rarely update their DQN weights by communicating with the central unit. Third, we are interested in both cooperative and non-cooperative settings, where fundamentally different operating points are reached depending on the network utility function. Furthermore, in \cite{challita2017proactive} the focus was on matching channels to base stations, whereas in our setting we focus on sharing the limited spectrum by a large number of users (i.e., matching might be infeasible). Other related work considered radio control and signal detection problems, in which a radio signal search environment based on Gym Reinforcement Learning was developed \cite{oshea2016deep} to approximate the cost of search, as opposed to asymptotically optimal search strategies \cite{cohen2015active, cohen2015asymptotically, huang2018active}. Other related works on the general topic of deep learning in mobile and wireless networking can be found in a very recent comprehensive survey \cite{zhang2018deep}.

\subsection{Other Related Work}
\label{ssec:related}

Related works on learning algorithms for DSA have mainly focused on model-dependent settings or myopic objectives so that tractable and structured solutions can be obtained. The problem has been widely studied under multi-armed bandit (MAB) formulations (and variations), in which the channels are represented as arms that the user aims to explore to receive high rewards (e.g., rates). Related works can be found in \cite{Ahmad_2009_Optimality, Liu_2010_Indexability, Wang_2012_Optimality, cohen2014restless} (and references therein) under the Bayesian setting and in \cite{Tekin_2012_Approximately, Liu_2013_Learning, avner2016multi, gafni2018learningISIT, gafni2018learning} (and references therein) under the non-Bayesian settings.
Another set of related work on the multi-user case was studied from game theoretic and congestion control (\cite{han2005fair, menache2008rate, candogan2009competitive, menache2011network, law2012price, wu2013fasa, singh2016combined, cohen2013game, cohen2016distributedToN, cohen2017distributed} and references therein), matching theory (\cite{leshem2012_multichannel, naparstek2014fully, naparstek2014distributed, gu2015matching, mochaourab2015distributed} and references therein), and graph coloring (\cite{wang2005list, wang2009improved, checco2014fast, checco2013learning} and references therein) perspectives.
Game theoretic aspects of the problem have been investigated from both non-cooperative (i.e., each user aims at maximizing an individual utility) \cite{menache2008rate, candogan2009competitive, singh2016combined, cohen2016distributedToN, cao2018distributed}, and cooperative (i.e., each user aims at maximizing a system-wide global utility) \cite{han2005fair, leshem2006bargaining, cohen2017distributed, bistritz2018approximate} settings. Matching algorithms have focused on allocating channels to users so that a certain utility is maximized (e.g., user sum rate) \cite{leshem2012_multichannel, naparstek2014fully, gu2015matching}. Graph coloring formulations have concerned with modeling the spectrum access problem as a graph coloring problem, in which users and channels are represented by vertices and colors, respectively. Thus, coloring vertices such that two adjacent vertices do not share the same color is equivalent to allocating channels such that interference between neighbors is being avoided (see \cite{wang2005list, wang2009improved, checco2014fast, checco2013learning} and references therein for related works).
Finally, all these studies mainly focused on model and objective-dependent problem settings, often require more complex implementations (e.g., carrier sensing, wideband monitoring), and the solutions are model-dependent and cannot effectively adapt in general for handling more complex real-world models.

\section{Network Model and Problem Statement}
\label{sec:network}

We consider a wireless network consisting of a set $\mathcal{N}=\left\{1, 2, ..., N\right\}$ of users and a set $\mathcal{K}=\left\{1, 2, ..., K\right\}$ of shared orthogonal channels (i.e., subbands). The users transmit over the shared channels using a random access protocol. At each time slot, each user is allowed to choose a single channel for transmission with a certain transmission probability (i.e., Aloha-type narrowband transmission). We assume that users are backlogged, i.e., all users always have packets to transmit. Transmission on channel $k$ is successful if only a single user transmits over channel $k$ in a given time slot. Otherwise, a collision occurs. Note that in the case where $N\leq K$, channel-user allocation is feasible, in which all users can always transmit and avoid collisions. The proposed algorithm in this paper applies to both $N\leq K$, and $N>K$ cases. After each time slot (say $t$), in which each user (say $n$) has attempted to transmit a packet, it receives a binary observation $o_n(t)$, indicating whether its packet was successfully delivered or not (i.e., ACK signal). If the packet has been successfully delivered, then $o_n(t)=1$. Otherwise, if the transmission has failed (i.e., a collision occurred), then $o_n(t)=0$.

Let \vspace{0.0cm}
\beq
\displaystyle a_n(t)\in\left\{0, 1, ..., K\right\} \vspace{0.0cm}
\eeq
be the action of user $n$ at time slot $t$, where $a_n(t)=0$ refers to the case in which user $n$ chooses not to transmit a packet at time slot $t$ (to reduce the congestion level for instance), and $a_n(t)=k$, where $1\leq k\leq K$, refers to the case in which user $n$ chooses to transmit a packet on channel $k$ at time slot $t$.
We define \vspace{0.0cm}
\beq
\displaystyle a_{-n}(t)=\left\{a_i(t)\right\}_{i\neq n} \vspace{0.0cm}
\eeq
as the action profile for all users except user $n$ at time slot $t$. We consider a distributed setting without online coordination or message exchanges between users used to manage the spectrum access. As a result, the network state at time $t$ (i.e., $a_{-n}(t)$) is only partially observed by user $n$ through the local signal $o_n(t)$.
The history $\mathcal{H}_n(t)$ of user $n$ at time $t$ is defined by the set of all actions and observations up to time $t$: \vspace{0.0cm}
\beq
\displaystyle \mathcal{H}_n(t)=\left(\left\{a_n(i)\right\}_{i=1}^{t},\left\{o_n(i)\right\}_{i=1}^{t}\right). \vspace{0.2cm}
\eeq

\begin{definition}
A strategy $\sigma_n(t)$ of user $n$ at time $t$ is a mapping from history $\mathcal{H}_n(t-1)$ to a probability mass function over actions $\left\{0, 1, ..., K\right\}$. The time series vector of strategies (or \emph{policy}) for user $n$ is denoted by $\boldsymbol{\sigma}_n=\left(\sigma_n(t), t=1, 2, ...\right)$. A strategy profile of all users except user $n$ is denoted by $\boldsymbol{\sigma}_{-n}=\left\{\boldsymbol{\sigma}_i\right\}_{i\neq n}$. A strategy profile of all users is denoted by $\boldsymbol{\sigma}=\left\{\boldsymbol{\sigma}_i\right\}_{i=1}^n$.
\vspace{0.2cm}
\end{definition}

For convenience, we often write strategy $\sigma_n(t)$ as a $1\times K$ row vector: \vspace{0.0cm}
\beq
\displaystyle\sigma_n(t)=\left(p_{n, 0}(t), p_{n, 1}(t), ..., p_{n, K}(t)\right), \vspace{0.0cm}
\eeq
where
\beq
\displaystyle p_{n, k}(t)=\Pr\left(a_n(t)=k\right), \vspace{0.0cm}
\eeq
is the probability that user $n$ takes action $a_n(t)=k$ at time $t$. Let $r_n(t)$ be a reward that user $n$ obtains at the beginning of time slot $t$. The reward depends on user $n$'s action $a_n(t-1)$ and other users' actions $a_{-n}(t-1)$ (i.e., the unknown network state that user $n$ aims to learn). The reward can be viewed as a function of the achievable data rate on the wireless channel (say channel $k$), i.e., $B\log_2\left(1+\mbox{SNR}_n(k)\right)$, where $B$ is the channel bandwidth, and SNR$_n(k)$ is the received SNR of user $n$ on channel $k$.
Let \vspace{0.0cm}
\beq
\label{eq:reward}
\displaystyle R_n=\sum_{t=1}^{T}\gamma^{t-1}r_n(t) \vspace{0.0cm}
\eeq
be the accumulated discounted reward, where $0\leq\gamma\leq 1$ is a discounted factor, and $T$ is the time-horizon of the game. We often set $\gamma=1$, or $\gamma<1$ when $T$ is bounded or unbounded, respectively. The objective of each user (say $n$) is to find a strategy $\boldsymbol{\sigma}_n$ that maximizes its expected accumulated discounted reward: \vspace{0.0cm}
\beq
\label{eq:objective}
\displaystyle\max_{\boldsymbol{\sigma}_n} \;\;
\textbf{E}\left[R_n(\boldsymbol{\sigma}_n, \boldsymbol{\sigma}_{-n})\right], \vspace{0.0cm}
\eeq
where $\textbf{E}\left[R_n(\boldsymbol{\sigma}_n, \boldsymbol{\sigma}_{-n})\right]$ denotes the expected accumulated discounted reward when user $n$ performs strategy $\boldsymbol{\sigma}_n$ and the rest of the users perform strategy profile $\boldsymbol{\sigma}_{-n}$.

\begin{remark}
It should be noted that we mainly focus on DSA in the open sharing model \cite{zhao2007survey}. Therefore, we often do not assume that there are primary and secondary users in the networks. Nevertheless, we can extend the model to handle these situations by adding external processes (i.e., which are not affected by other users' actions) to model the primary users activities. As a result, the network state that user $n$ aims at inferring at time $t$ is given by $(a_{-n}(t), a_{p}(t))$, where $a_{p}(t)$ is the action profile for all primary users at time $t$. In Section \ref{sec:sim}, we demonstrate strong performance of the proposed algorithm in the presence of primary users as well.
\end{remark}

We are interested in developing a model-free distributed learning algorithm to solve (\ref{eq:objective}) that can effectively adapt to topology changes, different objectives, different finite time-horizons (in which solving dynamic programming becomes very challenging, or often impossible for large $T$), etc. Computing an optimal solution, however, is a combinatorial optimization problem with partial state observations which is mathematically intractable as the network size increases \cite{zhao2007decentralized}. Thus, we adopt a DRL approach due to its capability to provide good approximate solutions while dealing with a very large state and action spaces. In the next paragraph we first describe the basic idea of Q-learning and DRL. We then develop the proposed algorithm that adopts a deep multi-user reinforcement learning approach for DSA design in Section \ref{sec:algorithm}.

\noindent
\textbf{Background on Q-learning and Deep Reinforcement Learning (DRL):} Q-learning is a reinforcement learning method that aims at finding good policies for dynamic programming problems. It has been widely applied in various decision making problems, primarily because its ability to evaluate the expected utility among available actions without requiring prior knowledge about the system model, and its ability to adapt when stochastic transitions occur \cite{watkins1992q}. The algorithm was originally designed for a single agent who interacts with a fully observable Markovian environment (in which convergence to the optimal solution is guaranteed under some regularity conditions in this case). It has been widely applied for more involved settings as well (e.g., multi-agent, non-Markovian environment) and demonstrated strong performance, although convergence to the optimal solution is open in general under these settings. Assume first that the network state $s_n(t)=a_{-n}(t)$ is fully observable by user $n$. By applying Q-learning to our setting, the algorithm updates a Q-value at each time $t$ for each action-state pair as follows: \vspace{0.0cm}
\beq
\bea{l}
\label{eq:Qlearning}
\displaystyle
Q_{t+1}\left(s_n(t), a_n(t)\right)=Q_t\left(s_n(t),a_n(t)\right)
\vspace{0.2cm}\\\hspace{1cm}\displaystyle
+\alpha\left[r_n(t+1)+\gamma\max_{a_n(t+1)}Q_t\left(s_n(t+1),a_n(t+1)\right)\right.
\vspace{0.2cm}\\\displaystyle\left.\hspace{5cm}
-Q_t\left(s_n(t),a_n(t)\right)\right], \vspace{0.0cm}
\ena
\eeq
where
\beq
\displaystyle r_n(t+1)+\gamma\max_{a_n(t+1)}Q_t\left(s_n(t+1),a_n(t+1)\right) \vspace{0.0cm}
\eeq
is the learned value obtained by getting reward $r_n(t+1)$ after taking action $a_n(t)$ in state $s_n(t)$, moving to next state $s_n(t+1)$, and then taking action $a_n(t+1)$ that maximizes the future Q-value seen at the next state. The term $Q_t\left(s_n(t),a_n(t)\right)$ is the old learned value. Thus, the algorithm aims at minimizing the Time Difference (TD) error between the learned value and the current estimate value. The learning rate $\alpha$ is set to $0\leq\alpha\leq 1$, where typically is set close to zero. When the problem is partially observable, the state is set to the history, i.e., $s_n(t)=\mathcal{H}_n(t)$ in our case (or a sliding window history when the problem size is too large). Throughout the paper we often remove the subscript $t$ to simplify the presentation.

While Q-learning performs well when dealing with small action and state spaces, it becomes impractical when the problem size increases for mainly two reasons: (i) A stored lookup table of $Q$-values for all possible state-action pairs is required which makes the storage complexity intolerable for large-scale problems. (ii) As the state space increases, many states are rarely visited, which significantly decreases performance.

In recent years, a great potential was demonstrated by DRL methods that combine deep neural network with Q-learning, referred to as Deep Q-Network (DQN), for overcoming these issues. Using DQN, the deep neural network maps from the (partially) observed state to an action, instead of storing a lookup table of Q-values. Furthermore, large-scale models can be represented well by the deep neural network so that the algorithm has the ability to preserve good performance for very large-scale models. Although convergence to the optimal solution of DRL is an open question (even for a single agent), strong performance has been demonstrated in various fields as compared to other approaches. A well known single-player DRL-based algorithm has been developed in DeepMind's recently published Nature paper \cite{mnih2015human}, for teaching computers how to play Atari games directly from the on-screen pixels, in which strong performance has been demonstrated in many tested games. For other recent developments see Section \ref{ssec:related}.

\section{The Proposed Deep Q-Learning for Spectrum Access (DQSA) Algorithm}
\label{sec:algorithm}

Direct computation of the optimal channel allocation and transmission probabilities for the multi-channel spectrum access problem (\ref{eq:objective}) is a combinatorial optimization problem with partial state observations which is mathematically intractable as the network size increases \cite{zhao2007decentralized}. Furthermore, it requires online centralized computations. Iterative algorithms that approximate (\ref{eq:objective}) have been mainly developed for specific problem settings, where obtaining a global network utility generally requires message exchanges between users (e.g., \cite{cohen2017distributed}). In this section, we develop the proposed DQSA algorithm based on deep multi-user reinforcement learning to solve (\ref{eq:objective}). The DQSA algorithm applies for general large and complex settings and does not require online coordination or message exchanges between users.

We first present in Section \ref{ssec:architecture} the proposed architecture of the DQN used in the DQSA algorithm. In Section \ref{ssec:training} we present the offline algorithm used for training the DQN, and in Section \ref{ssec:online} we describe the online learning algorithm for the distributed random access, in which every user operates in a fully distributed manner by using the trained DQN. The specific setting of the objective function used for training the DQN depends on the desired performance as will be discussed in Section \ref{sec:analysis}. Specifically, in Section \ref{sec:analysis} we establish design principles for implementing DQSA based on a game theoretic analysis of the operating points of (\ref{eq:objective}) under both cooperative and competitive utility functions.

\subsection{Architecture of the Proposed Multi-User DQN Used in DQSA Algorithm}
\label{ssec:architecture}
In this section, we describe the proposed architecture for the multi-user DQN used in DQSA algorithm to solve the DSA problem. An illustration of the DQN is presented in Fig. \ref{figure_nnet}.

\noindent
  \emph{1) Input Layer:} The input $\textbf{x}_n(t)$ to the DQN is a vector of size $2K+2$. The first $K+1$ input entries indicate the action (i.e., selected channel) taken at time $t-1$. Specifically, if the user has not transmitted at time slot $t-1$, the first entry is set to $1$ and the next $K$ entries are set to $0$. If the user has chosen channel $k$ for transmission at time $t-1$ (where $1\leq k\leq K$), then the $(k+1)^{th}$ entry is set to $1$ and the rest $K$ entries are set to $0$. The following $K$ input entries are the capacity of each channel (i.e., the packet transmission rate over a channel conditioned on the event that the channel is free, which is proportional to the channel bandwidth). The last input is $1$ if ACK signal has been received. Otherwise, if transmission has failed or no transmission has been executed, it is set to $0$.\vspace{0.0cm}\\

\begin{figure}[htbp]
\centering \epsfig{file=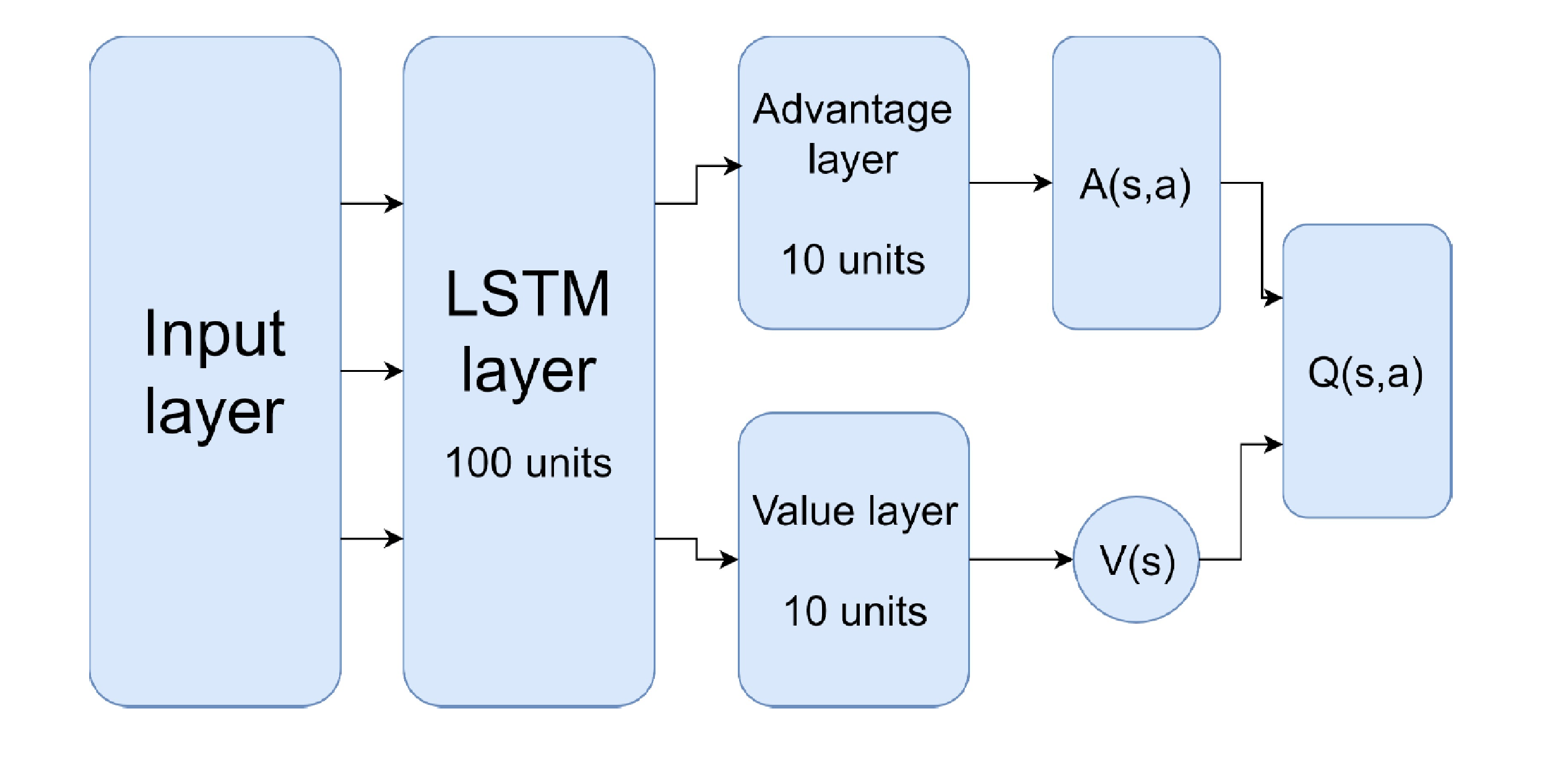,
width=0.48\textwidth}
\caption{An illustration of the architecture of the proposed multi-user DQN used in DQSA algorithm.}
\label{figure_nnet}
\end{figure}

\noindent
  \emph{2) LSTM Layer:} Since the network state is partially observable for each user, and the dynamics is non-Markovian and determined by the multi-user actions, classical DQNs do not perform well in this setting. Thus, we add an LSTM layer (~\cite{hausknecht2015deep}) to the DQN that maintains an internal state and aggregate observations over time. This gives the network the ability to estimate the true state using the history of the process. This layer is responsible of learning how to aggregate experiences over time.\vspace{0.2cm}\\
  \emph{3) Value and Advantage Layers:} Another improvement that we incorporate is the use of dueling DQN, as suggested in \cite{wang2015dueling}. The intuition behind this architecture lies in the fact that there is an observability problem in DQN. There are states which are good or bad regardless of the taken action. Hence, it is desirable to estimate the average Q-value of the state which is called the value of the state $V(s_n(t))$ independently from the advantage of each action. Thus, when we input $\textbf{x}_n(t)$ to the DQN with dueling, the Q-value for selecting action $a_n(t)$ at time $t$ is updated by: \vspace{0.0cm}
\beq
\displaystyle Q(a_n(t))\leftarrow V+A(a_n(t)). \vspace{0.0cm}
\eeq
Note that both $V$ and $A(a_n(t))$ depend on the state $s_n(t)$ (which is hidden and mapped by the DQN from the history). The term $V$ is the value of the state and it estimates the expected Q-value of the state with respect to the taken action. The term $A(a_n(t)$ is the advantage of each action and it estimates the Q-value minus its expected value. In practice, one way to evaluate $A(a_n(t)$ is to subtract the maximal value of the state with respect to the taken actions from the Q function. Another way is to subtract the average value of the state with respect to the taken actions from the Q function. Here, we use the latter method \cite{wang2015dueling}.
\vspace{0.2cm}\\
  \emph{4) Block output layer:} The output of the DQN is a vector of size $K+1$. The first entry is the estimated Q-value if the user will choose not to transmit at time $t$. The $(k+1)^{th}$ entry, where $1\leq k\leq K$, is the estimated Q-value for transmitting on channel $k$ at time $t$.\vspace{0.2cm}\\
  \emph{5) Double Q-learning:} The max operator in standard Q-learning and DQN (see (\ref{eq:Qlearning})) uses the same values to both selecting and evaluating an action. Thus, it tends to select overestimated values which degrades performance. Hence, when training the DQN, we use double Q-learning \cite{van2016deep} used to decouple the selection of actions from the evaluation of Q-values. Specifically, we use two neural networks, referred to as DQN$_1$ and DQN$_2$. DQN$_1$ is used for choosing actions and DQN$_2$ is used to estimate the Q-value associated with the selected action.

\subsection{Training the DQN:}
\label{ssec:training}

The DQN is trained for all users at a central unit in an offline manner.
We train the DQN as follows:\vspace{-0.2cm}

\noindent
\underline{\hspace{8.85cm}}\\
\underline{\textbf{DQSA Algorithm:} Training Phase\hspace{3.8cm}}
\begin{enumerate}
  \item \textbf{for} iteration $i=1, ..., R$ \textbf{do}\vspace{0.1cm}
  \item \hspace{0.2cm} \textbf{for} episode $m=1, ..., M$ \textbf{do}\vspace{0.1cm}
  \item \hspace{0.4cm}     \textbf{for} time-slot $t=1, ..., T$ \textbf{do}\vspace{0.1cm}
  \item \hspace{0.6cm}     \textbf{for} user $n=1, ..., N$ \textbf{do}\vspace{0.1cm}
  \item \hspace{0.8cm}         Observe an input $\textbf{x}_n(t)$ and feed it into the neural
                               \hspace*{0.825cm} network DQN$_1$
  \item \hspace{0.8cm}         Generate an estimation of the Q-values $Q(a)$ for
                               \hspace*{0.77cm} all available actions $a\in\left\{0, 1, ..., K\right\}$ by the
                               \hspace*{0.82cm} neural network \vspace{0.1cm}
  \item \hspace{0.8cm}         Take action $a_n(t)\in\left\{0, 1, ..., K\right\}$ (according to
                               \hspace*{0.82cm} (\ref{eq:P_a})) and obtain a reward
                               $r_n(t+1)$\vspace{0.1cm}
  \item \hspace{0.8cm}         Observe an input $\textbf{x}_n(t+1)$ and feed it into both
                               \hspace*{0.82cm} neural networks DQN$_1$ and DQN$_2$
  \item \hspace{0.8cm}         Generate estimations of the Q-values $\widetilde{Q}_1(a)$ and\\
                                \hspace*{0.67cm}  $\widetilde{Q}_2(a)$, respectively, for all actions
                               $a\in\hspace*{0.82cm}\left\{0, 1, ..., K\right\}$ by the neural networks \vspace{0.1cm}
  \item \hspace{0.8cm}         Form a target vector for the training by replacing
                                \hspace*{0.82cm} the $a_n(t)$ entry by:\vspace{-0.3cm}
\begin{center}
$\bea{l}
\hspace*{0.82cm}\displaystyle Q(a_n(t))\leftarrow r_n(t+1)+\widetilde{Q}_2\left(\arg\max_a\left(\widetilde{Q}_1(a)\right)\right)\vspace{0.1cm}
\ena$
\end{center}
  \item \hspace{0.6cm} \textbf{end for}\vspace{0.1cm}
  \item \hspace{0.4cm} \textbf{end for}\vspace{0.1cm}
  \item \hspace{0.2cm} \textbf{end for}\vspace{0.1cm}
  \item Train DQN$_1$ with inputs $\textbf{x}$s and outputs $Q$s.\vspace{0.1cm}
  \item Every $\ell$ iterations set $Q_2\leftarrow Q_1$.\vspace{0.1cm}
  \item \textbf{end for}\vspace{-0.3cm}
\end{enumerate}
\underline{\hspace{8.85cm}}\vspace{0.1cm}

\noindent
In our experiments, we repeated the outer loop for several thousands iterations until convergence, and $\ell$ was set to $5$. Note that unlike \cite{mnih2013playing, mnih2015human}, in which experience replay was used in the single-agent case to learn from past observations, in the multi-user case considered here such learning is undesirable due to interactions among users. Hence, we collect the $M$ episodes at each iteration and create target values for all the episodes.

\subsection{Online Learning: Distributed Random Access using DQN:}
\label{ssec:online}

The training phase is rarely required to be updated by the central unit (only when the environment characteristics have been significantly changed and no longer reflects the training experiences). Users' SDRs only need to update their DQN weights by communicating with the central unit. In real-time, each user (say $n$) makes autonomous decisions in online and distributed manners using the trained DQN, to learn efficient spectrum access policies from its ACK signals only: \vspace{0.2cm}\\
1) At each time slot $t$, obtain observation $o_n(t)$ and feed input $\textbf{x}_n(t)$ to the trained DQN$_1$. Output Q-values $Q(a)$ are generated by DQN$_1$ for all available actions $a\in\left\{0, 1, ..., K\right\}$. \vspace{0.2cm}\\
2) Play strategy $\sigma_n(t)$ as follows: Draw action $a_n(t)$ according to the following distribution: \vspace{0.0cm}
\beq
\bea{l}
\label{eq:P_a}
\displaystyle \Pr\left(a_n(t)=a\right)=\frac{\left(1-\alpha \right)e^{\beta Q\left(a\right)}}{\displaystyle\sum_{\tilde{a}\in\left\{0, 1, ..., K\right\}} e^{\beta Q\left(\tilde{a}\right)}}+\frac{\alpha}{K+1}\vspace{0.0cm}\\\hspace{4.5cm}
\displaystyle \forall a\in\left\{0, 1, ..., K\right\}, \vspace{0.0cm}
\ena
\eeq
for small $\alpha>0$, and $\beta$ is the temperature. Note that (\ref{eq:P_a}) balances between the softmax and $\epsilon$-greedy strategies, known as Exp3 strategy \cite{auer1995gambling}. In practice, $\alpha$ is small and we take it to zero with time, so that the algorithm becomes more greedy with time in terms of selecting actions with high estimated Q-values.
The game is played over a time-horizon of $T$ time slots.

\subsection{Computational Complexity:}
\label{ssec:computational}

The number of multiplications through the DQN with $G$ layers, in which $\tilde{K}$ is the size of the input layer which is proportional to the number of channels, and $d_g$ is the number of units in the $g$'th layer, is given by $D\triangleq\tilde{K}d_1+\sum_{g=1}^{G-1} d_g d_{g+1}$. Therefore, the computational complexity in real-time for each user is given by $O(D)$ at each time step. The expensive computational complexity is only done at the offline training phase. The computational complexity of the forward and back propagation for one sample is $O(D)$. The training complexity for one minibatch of $M$ episodes with $T$ time-steps and $N$ users is given by $O(MTND)$. This is done over $I$ iterations until convergence, which results in computational complexity of order $O(IMTND)$ in the training phase.

As explained and illustrated in Section \ref{ssec:main_results}, the proposed algorithm is very simple for implementation using simple SDRs. The expensive computations at the training phase can be done offline by a centralized powerful unit (e.g., using MEC settings), where updating the DQN is rarely required.

\section{Analysis of the System Dynamics with Different Utility Functions}
\label{sec:analysis}

Since users take autonomous actions when operating the spectrum access, it is convenient to model the network dynamics from a game theoretic perspective, which is used in this section. Since training the DQN with different objective functions might lead to significantly different operating points of the system, we are interested in establishing efficient design principles for the implementation of the DQSA algorithm. We investigate both non-cooperative and cooperative utilities of the system. We first define the Nash equilibrium point as a strategy profile for all users, in which there is no incentive for any user to unilaterally deviate from it. The users dynamics in this section is referred to as a multichannel random access game. \vspace{0.2cm}

\begin{definition}
A Nash equilibrium (NE) for the multichannel random access game is a strategy profile $\boldsymbol{\sigma}^*=(\boldsymbol{\sigma}_n^*,\boldsymbol{\sigma}_{-n}^*)$, such that
\beq
\label{eq:NEP}
\bea{l}
\displaystyle R_n\left(\boldsymbol{\sigma}_n^*,\boldsymbol{\sigma}_{-n}^*\right)\geq
\displaystyle R_n(\tilde{\boldsymbol{\sigma}}_n,\boldsymbol{\sigma}_{-n}^*)\;,\;\;\; \forall n \;, \forall \tilde{\boldsymbol{\sigma}}_n. \vspace{0.2cm}
\ena
\eeq
\end{definition}

A refinement of a NE is a subgame perfect equilibrium (SPE), which is a strategy profile that obeys a NE for each subgame.\vspace{0.2cm}
\begin{definition}
A subgame perfect equilibrium (SPE) for the multichannel random access game is a strategy profile $\boldsymbol{\sigma}^*$, if for any history $\left\{\mathcal{H}_n(t-1)\right\}_{n=1}^N$ for all $t$, the induced continuation strategy at times $t, t+1, ..., T$ is a NE of the continuation game that starts at time $t$ following history $\left\{\mathcal{H}_n(t-1)\right\}_{n=1}^N$. \vspace{0.2cm}
\end{definition}

NEs and SPEs describe operating points which are stable in terms of local efficiency. Specifically, no user has an incentive to unilaterally deviate from its current strategy given the current system state. However, these operating points might be highly inefficient in terms of the reward that users can obtain by cooperating. Thus, we next define efficient operating points in terms of Pareto optimality. \vspace{0.2cm}
\begin{definition}
A NE $\boldsymbol{\sigma}^*$ is Pareto-optimal if no strategy profile can improve the reward of one user without decreasing the reward of at least one other user. \vspace{0.2cm}
\end{definition}

Next, we analyze the operating points of the system under different utility functions. We will use this analysis for establishing design principles for the setting of DQSA algorithm used to bring the system to operate in efficient operating points.

\subsection{Competitive Reward Maximization}
\label{ssec:competitive}

The first optimization problem that we investigate is concerned with the case in which each user aims at maximizing its own rate. Specifically, let $\textbf{1}_n(t)$ be the indicator function, where $\textbf{1}_n(t)=1$ if user $n$ has successfully transmitted a packet at time slot $t$, and $\textbf{1}_n(t)=0$ otherwise. Let \vspace{0.0cm}
\beq
\label{eq:reward_individual}
\displaystyle r_n(t)=\textbf{1}_n(t-1). \vspace{0.0cm}
\eeq
As a result, by substituting (\ref{eq:reward_individual}) in (\ref{eq:reward}) each user (say $n$) aims to maximize the total number of its own successful transmissions (i.e., \emph{individual rate}). Next, we show that equilibrium points of competitive rate maximization are efficient when $N\leq K$, but might be highly inefficient when $N>K$. \vspace{0.2cm}

\begin{theorem}
\label{th:competitive_NE}
Set $r_n(t)$ as in (\ref{eq:reward_individual}).  Then , the following statements hold:\vspace{0.0cm}\\
1) Assume that $N\leq K$. Then, the following strategy profile is a SPE: (i) $p_{n,0}(t)=0\;\forall n, t$, (ii) $\sum_{k=1}^K p_{n,k}(t)=1 \;\forall n, t$, and (iii) for all $t$, if $p_{n,k}(t)>0$ for any $k$, then $p_{n',k}(t)=0$ for $n'\neq n$.
\vspace{0.0cm}\\
2) Assume that $N> K$, and assign channel $k_n$ for any user $n$, such that $k_n\in\left\{1, 2, ..., K\right\}$, and $\left\{1, 2, ..., K\right\}\subseteq\bigcup_{n=1}^N k_n$. Then, for any such assignment the following strategy profile is a SPE: (i) $p_{n,0}(t)=0\;\forall n, t$, (ii) $p_{n,k_n}(t)=1\;\forall n, t$. \vspace{0.0cm}
\end{theorem}
\begin{proof}
We start by proving the first statement. Note that the strategy profile described in the statement avoids collisions among users by transmitting on orthogonal channels at each time slot (condition (iii)). The strategy is feasible since $N\leq K$. By conditions (i) and (ii), each user surely receives $r_n(t)=1$ at each time slot. As a result, no user has an incentive to switch to a different channel or reduce its transmission probability since its individual rate will not increase. Since this argument holds for all $t$ independent of the history, the strategy profile described in Statement $1$ is a SPE.

Next, we prove the second statement. Since $N>K$, then there exists at least one channel $k\in\left\{1, ..., K\right\}$ which is assigned to at least two users (pigeonhole principle) under the strategy profile described in the statement. Let $\mathcal{K}_c$ be the set of all channels that are assigned to at least two users. Since $p_{n,k_n}(t)=1\;\forall n, t$, then:
\beq
\displaystyle \mbox{$R_n=0$, for all $n$ such that $k_n\in\mathcal{K}_c$.} \vspace{0.0cm}
\eeq
On the other hand,
\beq
\displaystyle \mbox{$R_n=\sum_{t=1}^{T}\gamma^{t-1}$ for all $n$ such that $k_n\nin\mathcal{K}_c$.} \vspace{0.0cm}
\eeq
Next, we show that no user has an incentive to switch to a different channel or reduce its transmission probability. Consider first user $n$ such that $k_n\in\mathcal{K}_c$. Since $\left\{1, 2, ..., K\right\}\subseteq\bigcup_{n=1}^N k_n$ by the condition (i.e., every channel is assigned to at least one user), and $p_{n,k_n}(t)=1\;\forall n, t$, then switching to a different channel or reducing its transmission probability, still results in getting $r_n(t)=0$ for all $t$. Next, consider user $n$ such that $k_n\nin\mathcal{K}_c$. Under the current strategy profile its individual reward is maximized, and it has no incentive to deviate from it. As a result, the strategy profile described in Statement $2$ is a SPE.
\end{proof}

Theorem \ref{th:competitive_NE} implies that when $N>K$, the equilibrium points of the system might be highly inefficient. In fact, any learning dynamics among users in which users can update sequentially their transmission probability to increase their individual rate, will result in increasing the transmission probability close to $1$ (since every user has an incentive to increase its rate by increasing the transmission probability as long as the channel yields a positive capacity).
To avoid the situation in which users keep increasing their transmission probability to increase their rates, we develop a mechanism that restricts their strategy space when training the DQN, as described below.\vspace{0.2cm}

\begin{definition}
We say that DQSA algorithm is implemented using a common training, when the Q values in the training phase (see Section \ref{ssec:training}) are estimated under the implicit assumption that all users use the same protocol rules, i.e., $\sigma_n(t)=\sigma_{n'}(t)$ for all $n, n'$, for all $t$.\vspace{0.2cm}
\end{definition}

The next proposition shows that implementing DQSA algorithm using a common training avoids convergence to competitive SPEs as described in Theorem \ref{th:competitive_NE} Statement 2. To avoid trivial solutions, it is assumed that users are allowed to transmit with probability $1-\epsilon$ for small $\epsilon>0$ (otherwise if users always transmit with probability $1$, the reward equals zero on all channels).

\begin{proposition}
\label{prop:common}
Fix $K$, and assume that DQSA algorithm is implemented using a common training. Then, the probability that the algorithm will converge to competitive SPEs approaches zero as $N$ approaches infinity.\vspace{0.0cm}
\end{proposition}

\begin{proof}
Under any competitive SPE in Theorem \ref{th:competitive_NE} Statement 2 (when users are allowed to transmit with probability $1-\epsilon$ for small $\epsilon>0$), every user transmits on a single channel with probability $1-\epsilon$ at each given time. Let $N_k(t)$ be the number of users that transmit on channel $k$ at time $t$. As $N$ approaches infinity, $N_k(t)$ approaches $N/K$. Otherwise, users have an incentive to switch channels. As a result, the reward for each user approaches $(1-\epsilon)\epsilon^{N/K-1}$ for all $t$ which approaches zero exponentially fast with $N$. On the other hand, the reward for each user when applying a simple strategy in which every user transmits over a randomly selected channel with probability $K/N$ approaches $Ke^{-1}/N$. Thus, when applying a common training the Q values increase when decreasing the transmission probabilities as $N$ increases, which avoids convergence to competitive SPEs.\vspace{0.0cm}
\end{proof}

Proposition $1$ establishes an important design principle. It implies that implementing DQSA algorithm using a common training avoids the algorithm to reach highly inefficient operating points. Next, we characterize the Pareto optimal operating points of the system when $N>K$ (i.e., when SPEs are inefficient).

\begin{theorem}
\label{th:competitive_Pareto}
Assume that $N>K$ and set $r_n(t)$ as in (\ref{eq:reward_individual}). Then , the following strategy profile is Pareto optimal: for each time $t$, for every channel $k\in\left\{1, ..., K\right\}$ there exists a user $n_k(t)$, such that $p_{n_k(t),k}(t)=1$ and $p_{n',k}(t)=0$ for all $n'\neq n_k(t)$.
\vspace{0.0cm}
\end{theorem}
\begin{proof}
Let $\boldsymbol{\sigma}^*$ be the strategy profile defined by the theorem. Let $\boldsymbol{\sigma}'$ be a strategy profile in which user $n$ gets higher reward: $R_n(\boldsymbol{\sigma}')>R_n(\boldsymbol{\sigma}^*)$. We define the total reward for all users under any strategy profile $\boldsymbol{\sigma}$ by $S_R(\boldsymbol{\sigma})=\sum_{n=1}^N R_n(\boldsymbol{\sigma})$. Since there are no collisions under $\boldsymbol{\sigma}^*$, then the total reward for all users under $\boldsymbol{\sigma}^*$ is given by:
\beq
\displaystyle S_R(\boldsymbol{\sigma}^*)=K\sum_{t=1}^{T}\gamma^{t-1}. \vspace{0.0cm}
\eeq
Next, since $S_R(\boldsymbol{\sigma}^*)\geq S_R(\boldsymbol{\sigma}')$ and $R_n(\boldsymbol{\sigma}')>R_n(\boldsymbol{\sigma}^*)$, the total rewards for all users except user $n$ under $\boldsymbol{\sigma}^*$, and $\boldsymbol{\sigma}'$ satisfy:
\beq
\displaystyle S_R(\boldsymbol{\sigma}^*)-R_n(\boldsymbol{\sigma}^*)>S_R(\boldsymbol{\sigma}')-R_n(\boldsymbol{\sigma}'). \vspace{0.0cm}
\eeq
Hence, there exists a user $n'$ that receives a smaller reward when the system switches from $\boldsymbol{\sigma}^*$ to $\boldsymbol{\sigma}'$, $R_{n'}(\boldsymbol{\sigma}')<R_{n'}(\boldsymbol{\sigma}^*)$. Hence, $\boldsymbol{\sigma}^*$ is Pareto optimal.
\vspace{0.0cm}
\end{proof}

Theorem \ref{th:competitive_Pareto} implies that any strategy profile that shares resources without collisions among users is Pareto optimal. In Section \ref{sec:sim}, we implemented DQSA algorithm using a common training, and it is shown that users indeed avoid inefficient SPEs (as stated in Proposition \ref{prop:common}). Interestingly, it is shown that the users often reach (in about $80\%$ of the Monte-Carlo experiments) Pareto optimal strategies as characterized by Theorem \ref{th:competitive_Pareto} using only ACK signals. Although convergence of DRL to optimal strategies is an open question, the intuition for reaching Pareto optimal strategies can be explained as follows. Assume that a user has succeeded to learn well the system state from its history using the DQN (which occurs often since large-scale partially observed models can be represented well by the DQN). Since users use common training when updating their strategy, they tend to strategies that avoid collisions to increase the reward. Which one of the operating points is reached depends on the initial conditions and randomness of the algorithm.

\subsection{Cooperative Reward Maximization}
\label{ssec:cooperative}

In this section, we investigate the case in which every user in the system aims at maximizing the same global system-wide reward. Specifically, let
\vspace{0.0cm}
\beq
\label{eq:reward_cooperative_t}
\displaystyle \mbox{$r_n(t)=0$, for all $1\leq t\leq T-1$,} \vspace{0.0cm}
\eeq
and \vspace{0.0cm}
\beq
\label{eq:reward_cooperative_T}
\displaystyle r_n(T)=\sum_{n=1}^{N}f\left(\sum_{t=1}^{T}\textbf{1}_n(t-1)\right). \vspace{0.0cm}
\eeq
The function $f(x)$ can be designed so as to achieve a certain network utility. We focus here on the unified system-wide $\alpha$-fair utility function which is given by \cite{srikant2013communication}:
\beq
\label{eq:alpha_fair}
\displaystyle f(x)=\frac{x^{1-\alpha}}{1-\alpha},
\;\;\; \mbox{for}\;\;\; \alpha\geq 0.
\eeq

It should be noted that various well-known system-wide utility functions are special cases of the unified $\alpha$-fair utility function. For example, setting $\alpha=0$ results in maximizing the user sum-rate (since $f\left(x\right)=x$). Setting $\alpha=1$ results in maximizing the user sum log-rate, which is known as proportional fairness \cite{kar2004achieving} (since differentiating $f(x)-Const$, where $Const=1/(1-\alpha)$ and taking the limit as $\alpha\rightarrow 1$ yields $f\left(x\right)=\log(x)$).

Next, we characterize the operating points of the system under the cooperative utility function, which are fundamentally different from the operating points under the competitive reward setting.\vspace{0.2cm}

\begin{theorem}
\label{th:cooperative}
Set $r_n(t)$ as in (\ref{eq:reward_cooperative_t}), (\ref{eq:reward_cooperative_T}), (\ref{eq:alpha_fair}). Then , the following statements hold:\vspace{0.0cm}\\
1) Assume that $\alpha=0$ in (\ref{eq:alpha_fair}). Then, the following strategy profile is SPE and Pareto optimal: for each time $t$, for every channel $k\in\left\{1, ..., K\right\}$ there exists a user $n_k(t)$, such that $p_{n_k(t),k}(t)=1$ and $p_{n',k}(t)=0$ for all $n'\neq n_k(t)$.
\vspace{0.0cm}\\
2) Assume that $\alpha>0$ in (\ref{eq:alpha_fair}) and $KT/N\in\mathbb{N}$. Then, the following strategy profile is SPE and Pareto optimal: (i) for each time $t$, for every channel $k\in\left\{1, ..., K\right\}$ there exists a user $n_k(t)$, such that $p_{n_k(t),k}(t)=1$ and $p_{n',k}(t)=0$ for all $n'\neq n_k(t)$. (ii) Each user transmits during $KT/N$ time slots, i.e., $\sum_{t=1}^{T}\textbf{1}_n(t)=KT/N$ for all $n$.
\vspace{0.0cm}
\end{theorem}
\begin{proof}
Let $x_n\triangleq\sum_{t=1}^{T}\textbf{1}_n(t-1)$. For proving both statements we first solve the following optimization problem:
\beq
\label{opt_proof}
\bea{l}
\displaystyle\max\;\;\sum_{n=1}^{N}\gamma^{T-1}\frac{x_n^{1-\alpha}}{1-\alpha}\;,
\displaystyle \;\;\;\mbox{s.t.}\;\;\;\; \sum_{n=1}^{N}x_n\leq KT.
\ena
\eeq
Note that (\ref{opt_proof}) maximizes the total reward that each user can get subject to constraint on the total number of transmissions in the network, which equals $KT$. The Lagrangian for the problem is given by:
\beq
\displaystyle L(\textbf{x},\lambda)=\sum_{n=1}^{N}\gamma^{T-1}\frac{x_n^{1-\alpha}}{1-\alpha}
-\lambda\left(\sum_{n=1}^{N}x_n-KT\right),
\eeq
for $\lambda\geq 0$. Differentiating with respect to $x_n$ yields $x_n^{-\alpha}=\lambda$ for all $n$. As a result, when $\alpha>0$, we have $x_1=x_2=\cdots=x_N=KT/N$.
When $\alpha=0$, we have $\sum_{n=1}^{N}x_n=KT$, so that any partition of $KT$ among users solves (\ref{opt_proof}).

Next, we prove the statements. We first prove Statement $1$. Since the strategy profile defined in Statement $1$ avoids collisions, then it satisfies the solution to (\ref{opt_proof}) under $\alpha=0$. Since any unilaterally deviation by a single user results in collisions, no user has an incentive to deviate at each subgame. Thus, the strategy profile is SPE. Also, we cannot increase the reward of any user by switching to another strategy profile (since it solves (\ref{opt_proof})). Thus, the strategy profile is also Pareto optimal.

Next, we prove Statement 2. Since the strategy profile defined in Statement $2$ avoids collisions and also partitions the time slots equally among users, then it satisfies the solution to (\ref{opt_proof}) under $\alpha>0$. Since any unilaterally deviation by a single user results in collisions, no user has an incentive to deviate at each subgame (although the total reward at each subgame (say at the remaining time slots $t_s+1, ..., T$) might not be optimal for the subgame since $\sum_{t=t_s+1}^{T}\textbf{1}_n(t)$ does not necessarily equal $K(T-t_s)/N$). Thus, the strategy profile is SPE. Also, we cannot increase the reward of any user by switching to another strategy profile under the total game (played at time slots $t=1, 2, ..., T$) since it solves (\ref{opt_proof}). Hence, the strategy profile is also Pareto optimal.\vspace{0.2cm}
\end{proof}

\begin{remark}
Theorem \ref{th:cooperative} implies that when $\alpha=0$, any strategy profile that avoids collisions and idle time slots is Pareto optimal and SPE. In Section \ref{sec:sim}, we trained the DQN with $\alpha=0$ (i.e., for maximizing the user sum rate). We observed that the proposed DQSA algorithm often reaches these strategies by learning from ACK signals only. Interestingly, the algorithm often converges to the simplest form of these strategies, in which only a subset of the users transmit for all $t=1, ..., T$. On the other hand, when $\alpha>0$, Theorem \ref{th:cooperative} implies that any strategy profile that avoids collisions and idle time slots, and also equally shares the time slots among users is Pareto optimal and SPE. In Section \ref{sec:sim}, we trained the DQN with $\alpha=1$ (i.e., for maximizing the user sum log-rate). We observed that the proposed DQSA algorithm often reaches these strategies as well by learning from ACK signals only. Although convergence of DRL to optimal strategies is an open question, the intuition for often reaching the desired strategies can be explained as follows. Assume that a user has succeeded to learn well the system state from its history using the DQN (which occurs often since large-scale partially observed models can be represented well by the DQN). Since all users receive the same global reward, they aim at maximizing the same global function (or potential function). Thus, selecting actions with high temperature in (\ref{eq:P_a}) converges to an operating point that maximizes the reward, resulting in Pareto optimal strategies according to Theorem \ref{th:cooperative}.
\end{remark}

\subsection{Maximizing a Global Utility with Aggregated Rewards}
\label{ssec:achieving}

When directly optimizing a global system-wide fairness utility, the reward for each user is no longer aggregated over time and depends on the common global utility, which is received by time $T$ (i.e., the end of the episode). However, it is well known that receiving delayed rewards decreases the training efficiency, and in our case the delay is the total time horizon. For handling this challenge, we exploit the inherent structure of the objective function to design a reward which is aggregated over time and approximates well the system-wide global utility when training the DQN. When the objective is the sum rate, this can be implemented by adding the sum of successfully transmitted packets at each given time for each user. When the objective is the sum log-rate (i.e., proportional fairness criterion), we use the harmonic number $H_n=\sum_{l=1}^n\frac{1}{l}$ as an approximation to $\log(n)$. From
\cite[pp. 73-75]{havil2003exploring}, we know that $\frac{1}{2(n+1)}<H_n-\log(n)-\gamma<\frac{1}{2n},$
where $\gamma$ is the Euler-Mascheroni constant. The bounds become tight for large $n$.
We define $M_n(t)$ as the number of successful transmissions by user $n$ until time $t$:
$
M_n(t)\triangleq\sum_{l=1}^t\textbf{1}_n(l).
$
Then, we can write:
\beq
\bea{l}
\displaystyle\sum_{n=1}^N\log(M_n(T)) \approx \sum_{n=1}^N \sum_{t=1}^{T}\frac{1}{M_n(t)}\textbf{1}_n(t) \vspace{0.1cm}\\ \hspace{3cm}
\displaystyle=\sum_{t=1}^{T}\sum_{n=1}^N \frac{1}{M_n(t)}\textbf{1}_n(t),
\ena
\eeq
where we used $\sum_{t=1}^{T}\frac{1}{M_n(t)}\textbf{1}_n(t)=\sum_{m=1}^{M_n(T)}\frac{1}{m}$ to replace the logarithm by the harmonic number. As a result, we obtain that every successful transmission by user $n$ at time $t$ contributes $1/M_n(t)$ to the total utility.
Using the above approximation, we can define the modified reward for the proportional fairness criterion as
\beq
r_n(t)\triangleq\sum_{n=1}^N\frac{1}{M_n(t)}\textbf{1}_n(t),
\eeq
for all $n=1, ..., N$, for all $t=1, ..., T$. Note that we use this modified reward only at the centralized training pase. In real-time, each user makes autonomous decisions based on ACK signals only. Using the above modified reward significantly improves performance in terms of achieving proportionally fair rates as demonstrated in Section \ref{sec:sim}.

\section{Experiments}
\label{sec:sim}
In this section, we present numerical experiments to illustrate the performance of the proposed DQSA algorithm. The simulations were implemented in Matlab. We simulated a wireless network consisting of $N$ users and $K$ orthogonal channels, as described in Section \ref{sec:network}, where $N$ varies between $3$ and $100$, and $K$ varies between $2$ and $50$ for different experiment settings. We simulated Rayleigh fading channels, with SNR$=35$dB, and channel bandwidth $B=20$MHz, when computing the data rate. The DQN includes LSTM layer with $100$ units, and two duelling layers of $10$ units ($10$ for A and $10$ for V). The minibatch size was set to $16$ episodes of $50$ time steps each. The discount factor was set to $\gamma=0.95$. We set $\alpha$ to $0.05$ at the beginning of the training and decreased it slowly to $0$. We increased the temperature $\beta$ slowly from $1$ at the beginning of the training up to $20$. We trained the network over $10,000$ iterations. To reduce the training complexity of the DQN, each user selected a channel from a set of two channels under DQSA. After training the DQN, we tested performance by averaging over $1000$ experiments of $100-200$ time slots. All the reported results obtained in a distributed manner given the trained DQN for each user.

The channel throughput under DQSA was compared to the classical slotted-Aloha protocol in Section \ref{ssec:sim_channel_throughput}. In Section \ref{ssec:sim_utility} we examined the achievable rate under various random access algorithms. In addition to DQSA algorithm, we simulated the following algorithms for comparison: (i) Opportunistic Channel Aware (OCA) protocol that uses channel state information for exploiting the channel diversity and access the channel with the highest achievable rate (irrespective of the collision rate) \cite{to2010exploiting, cohen2016distributedToN}; and (ii) Distributed Protocol (DP), that uses distributed learning by Gibbs sampler when selecting channels and transmission probabilities to converge to (nearly) optimal proportionally fair rates \cite{hou2014proportionally}.

\subsection{Complexity Comparison}
\label{ssec:sim_complexity}

In terms of overhead complexity of the protocols, both DP and OCA algorithms require frequent message exchanges between users. Once a user updates its transmission parameters (i.e., selected channel and transmission probability), it sends this information to its neighbors. This information is used to update the transmission parameters of other users in future iterations. By contrast, DQSA learns good policies from ACK signals only, and does not require those message exchanges between users, which becomes an important advantage in terms of reducing the protocol overhead.

In terms of computational complexity, all algorithms require $O(K)$ computations at each time a user updates its transmission parameters. The constant factor is the smallest under the OCA algorithm, since a user simply runs over an unsorted array of size $K$ (i.e., data rate for each channel) when selecting the channel with the highest rate. Then, it updates the transmission probability based on the information received from other users. The DP algorithm is slightly more involved. The user first multiplies the rate of each channel by the packet success probability (based on the information received from other users), then maps the resulting $K$ values to probability mass function (i.e., Gibbs distribution) over the channels, and finally draws the selected channel from this probability mass function. Under DQSA, the constant factor is significantly higher due to passing the observed input through the DQN (see a detailed complexity analysis in Section \ref{ssec:computational}). In our simulations, passing the input through the DQN requires $(2K+2)\times 100$ multiplications (due to 100 units LSTM layer) plus $2\times 100\times 10$ multiplications (due to 10 units Value and Advantage layers). Then, the $K+1$ Q-values are mapped to probability mass function of the Exp3 strategy (\ref{eq:P_a}), and finally the action is drawn from this probability mass function. A discussion on current developments of mobile devices that support computationally intense deep learning algorithms is provided in Section \ref{sec:conclusion}.

\subsection{Learning to Increase the Channel Throughput}
\label{ssec:sim_channel_throughput}

Since there is no coordination between users, inefficient channel utilization occurs when no user accesses the channel (referred to as idle time slots) or whether two or more users access the channel at the same time slot (i.e., collisions). The channel throughput is the fraction of time that packets are successfully delivered over the channel, i.e., no collisions or idle time slots occur. We simulated a network with disconnected cliques with a random number of users distributed uniformly between $3$ and $11$. At each clique, transmission is successful if only a single user in the clique transmits over a shared channel in a given time slot. There is no interference between users located at different cliques (e.g., uplink communication with scattered hotspots).
In this scenario, we compared the following schemes:
(i) \emph{The slotted Aloha protocol with optimal transmission probability:} In this scheme each user at clique $j$ transmits with probability $p_j$ at each time slot. Aloha-based protocols are widely used in wireless communication primarily because of their ease of implementation and their random nature. Setting $p_j=1/n_j$ is known to be optimal from both fairness (proportional fairness \cite{kar2004achieving}, max-min fairness) and Nash bargaining \cite{cohen2016distributedToN} perspectives. We assume that users set their transmission probability to the optimal value $p_j=1/n_j$ and computed the expected performance analytically as a benchmark for comparison.
(ii) \emph{The proposed DQSA algorithm:} We implemented the proposed algorithm, in which each user has the freedom to choose any transmission probability at each time slot. We implemented the DQSA algorithm by the competitive, sum-rate, and sum log-rate objectives as detailed in Section \ref{sec:analysis}.

We are interested to address the following question: Under slotted-Aloha, the expected channel throughput (conditioned on $n_j$) is given by $n_j p_j (1-p_j)^{n_j-1}=(1-1/n_j)^{n_j-1}\in(0.385,0.45)$ for $3\leq n_j\leq 11$ and decreases to $e^{-1}\approx 0.37$ as $n_j$ increases. \emph{We are thus interested to examine whether the users can effectively learn in a fully distributed manner only from their ACK signals how to access the channel so as to increase the channel throughput by reducing the number of idle time slots and collisions.} To make this question  more challenging, the actual number of users at each clique was unknown to the users when implementing the proposed algorithm (in contrast to the implementation of slotted-Aloha).

Figure \ref{figure_aloha_episode} provides a positive answer to this question for the experiments that we did. We point out that we do not use any coordination between users, message exchanges, etc. Therefore, the proposed algorithm starts from an aggressive strategy of frequent transmissions by the users to explore the system states, which is highly suboptimal. As a result, the performance improves drastically in the beginning of the algorithm due to the learning process, and significantly outperforms the slotted-Aloha protocol very quickly. \emph{The algorithm was able to deliver packets successfully almost $80\%$ of the time, about twice the channel throughput as compared to slotted-Aloha with optimal transmission probability. This is achieved when each user learns only from its ACK signals, without online coordination, message exchanges between users, or carrier sensing.} We point out that DQSA achieved almost $0.8$ channel throughput under the competitive reward, sum rate, and proportionally fair rates criterions. Thus, for the simplicity of presentation we present the DQSA performance under the competitive reward criterion in Fig. \ref{figure_aloha_episode}.

\begin{figure}[htbp]
\centering \epsfig{file=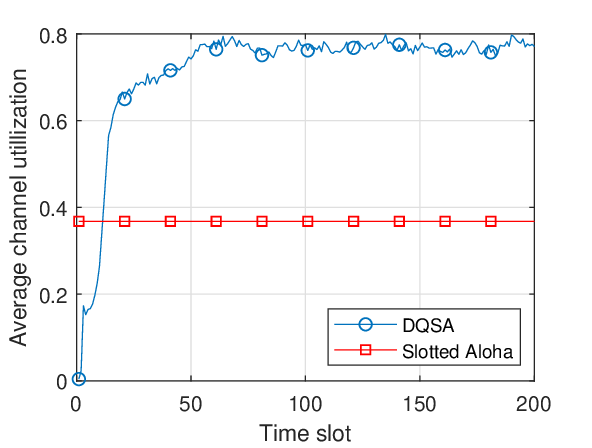,
width=0.48\textwidth}
\caption{Channel throughput for the experiments conducted in Section \ref{ssec:sim_channel_throughput}.}
\label{figure_aloha_episode}
\end{figure}

\subsection{Algorithm comparison Using Different Utility Functions}
\label{ssec:sim_utility}

Channel throughput is an important measure for communication efficiency, but it does not provide an indication about the desired performance among users. For example, if user $1$ transmits $100\%$ of the time and all other users receive rate zero, then the channel throughput is $1$, but the solution might be undesirable. Hence, in this section we are interested to address the following question: \emph{Can we train the DQN by different utilities so that the users can learn policies that result in good rate allocations depending on the desired performance?} In what follows, we provide a positive answer to this question for the experiments that we did.

We first simulated the case where $100$ users share $50$ channels. In Fig. \ref{fig:avg_rate} we also incorporated maximal Doppler shift of $100$Hz. It can be seen in \ref{fig:avg_rate} that DQSA algorithm achieves strong performance in terms of average user rate under all objective functions as desired. In Fig. \ref{fig:avg_log_rate} we present the average log-rate under various algorithms to measure performance in terms of the proportional fairness criterion. It can be seen that DQSA algorithm achieves the best performance under the proportional fairness objective, as desired, and outperforms both DP and OCA algorithms. Note that DQSA achieves poor average log-rate under the competitive and sum-rate criterion as desired since those objectives do not aim to achieve proportionally fair rates.

\begin{figure}[htbp]
\centering \epsfig{file=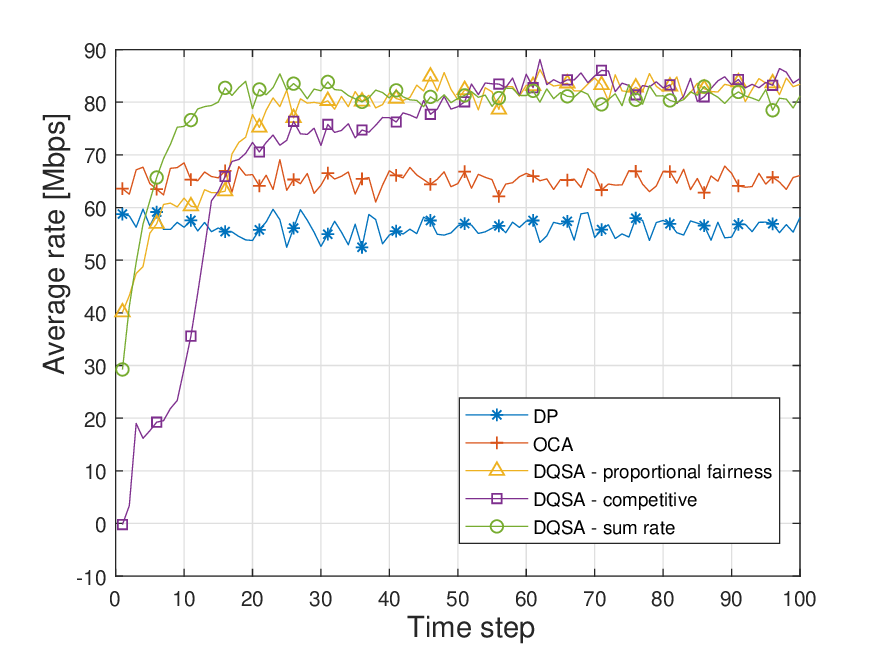,
width=0.48\textwidth}
\caption{Average user rate as a function of time under various algorithms. A case of $100$ users that share $50$ channels.}
\label{fig:avg_rate}
\end{figure}

\begin{figure}[htbp]
\centering \epsfig{file=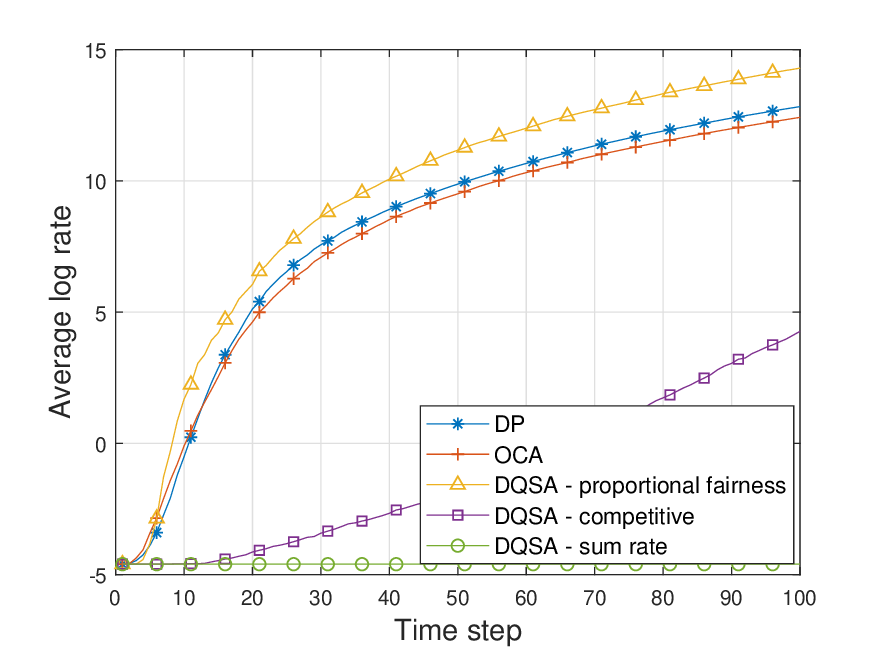,
width=0.48\textwidth}
\caption{Average user log-rate as a function of time (i.e., proportionally fair rates) under various algorithms. A case of $100$ users that share $50$ channels.}
\label{fig:avg_log_rate}
\end{figure}

We next investigate the case where $100$ secondary users and $50$ primary users share $50$ channels. The primary users activity is modeled by Markovian processes as commonly assumed in the literature \cite{zhao2007decentralized, Ahmad_2009_Optimality, Liu_2010_Indexability, Wang_2012_Optimality, Tekin_2012_Approximately, Liu_2013_Learning, cohen2014restless, gafni2018learningISIT}. Specifically, we assume that each channel follows an external ON/OFF Markov process due to primary users activities, i.e., each channel is ON when a primary user does not transmit on it, or OFF otherwise, with a stable probability to be ON set to $0.5$. Fig. \ref{fig:primary} shows that DQSA algorithm significantly outperforms the other algorithms in this scenario as well.

\begin{figure}[htbp]
\centering \epsfig{file=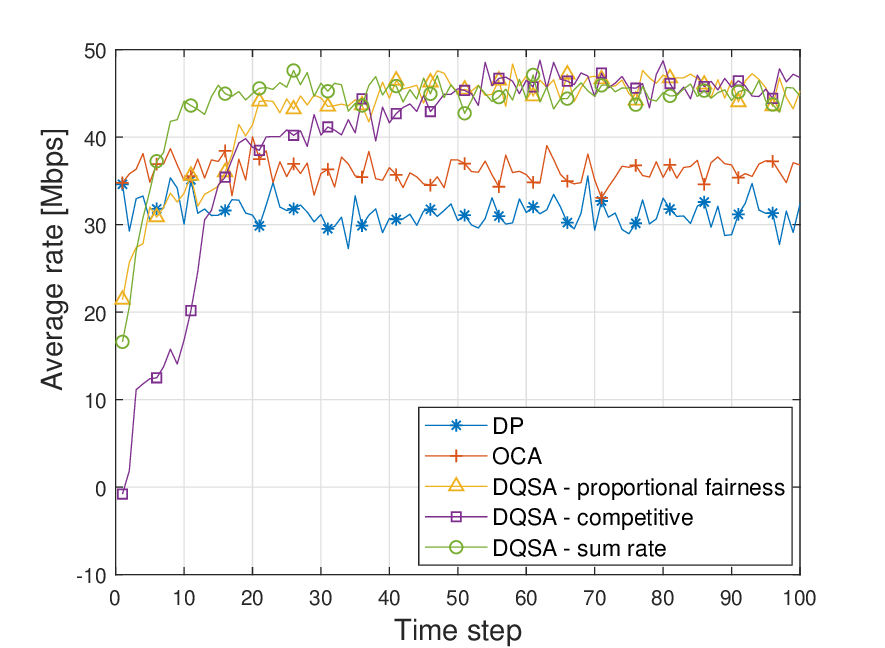,
width=0.48\textwidth}
\caption{Average secondary user rate as a function of time under various algorithms. A case of $100$ secondary users and $50$ primary users that share $50$ channels.}
\label{fig:primary}
\end{figure}

The experimental results support the following insights:

(i) The performance of DQSA algorithm under the competitive and proportional fairness objectives are clearly better than the performance of DQSA algorithm under the user sum-rate objective from a fairness perspective. This result is desirable since maximizing the user sum rate can be achieved by letting $K$ users to always transmit over $K$ channels, and the rest $N-K$ users receive zero rate. Indeed, these types of solutions perform poorly from a fairness perspective, since the sum log-rate tends to minus infinity. Under the competitive reward, however, each user tries to reach a good operating point so as to maximize its own rate. Which one of the users succeeds better is affected by the initial conditions and randomness of the algorithm. Therefore, an improvement in performance from a fairness perspective is expected, as observed in Fig. \ref{fig:avg_log_rate}. Under the proportional fairness objective, the users aim to equally share the channels for maximizing the objective function, as supported by Theorem \ref{th:cooperative}, and the high average log-rate demonstrated in fig. \ref{fig:avg_log_rate}.

(ii) From a game theoretic perspective, the SPE of the competitive game is reached when each user transmits with probability $1$ at each time slot, as shown by Theorem \ref{th:competitive_NE}, which is highly inefficient. Thus, implementing the DQSA algorithm using a common training yields a tremendous improvement in this respect, as can be seen in Figs. \ref{figure_aloha_episode}, \ref{fig:avg_rate}.

(iii) Finally, in about $80\%$ of the Monte-Carlo experiments we observed that DQSA algorithm converged to Pareto optimal resource sharing solutions, as analyzed in Section \ref{sec:analysis}. Specifically, under the sum rate objective, we often observed convergence to solutions in which only a subset of the users transmits during the entire time horizon. Since each user contributes equally to the user sum rate, the users often learn this simple and efficient policy that achieves this goal. This observation is demonstrated by high channel throughput in Fig. \ref{figure_aloha_episode} and high average rate in Fig. \ref{fig:avg_rate}, but poor average log-rate in Fig. \ref{fig:avg_log_rate}. Under the competitive reward objective, we often observed convergence to solutions in which collisions and idle time-slots are avoided, which is Pareto optimal. The users share the channels unequally but no user receives zero rate, due to the competitive nature of the reward. Which one of the users receives higher rate is affected by the initial conditions and randomness of the algorithm. This observation is demonstrated by high channel throughput in Fig. \ref{figure_aloha_episode}, high average rate in Fig. \ref{fig:avg_rate}, and a finite average log-rate in Fig. \ref{fig:avg_log_rate}. Under the proportional fairness objective, we often observed convergence to solutions in which collisions and idle time-slots are avoided, and users (nearly) equally share the channels during the time horizon, which is Pareto optimal. This observation is demonstrated by high channel throughput in Fig. \ref{figure_aloha_episode}, high average rate in Fig. \ref{fig:avg_rate}, and high average log-rate in Fig. \ref{fig:avg_log_rate}. These results demonstrate the strong performance of the DQSA algorithm and its capability to adapt to different problem settings.

\section{Conclusion}
\label{sec:conclusion}

The problem of dynamic spectrum access for network utility maximization in multichannel wireless networks was considered. We developed a novel distributed dynamic spectrum access algorithm based on deep multi-user reinforcement leaning, referred to as Deep Q-learning for Spectrum Access (DQSA). The proposed algorithm enables each user to learn good policies in an online and distributed manner, while dealing with the large state space without online coordination or message exchanges between users. Analysis of the system dynamics is developed for establishing design principles for the implementation of the DQSA algorithm. Experimental results demonstrated strong performance of the algorithm in complex multi-user scenarios.

It should be noted that the need for more efficient hardware acceleration of AI algorithms has been recognized by academia and industry in recent years, and currently big semiconductor companies, and startups develop chips for mobile devices that support computationally intense deep learning algorithms with low-power consumption. Hence, DRL-based algorithms has a great potential for providing effective solutions to real-world complex DSA challenges in practice. Future research direction that we intend to pursue in this respect is to develop creative hardware implementations for the proposed algorithm.

\bibliographystyle{ieeetr}

\end{document}

%% file: macros.tex
\newtheorem{theorem}{Theorem}
\newtheorem{lemma}{Lemma}
\newtheorem{remark}{Remark}
\newtheorem{definition}{Definition}

\newtheorem{proposition}{Proposition}

\newcommand{\beq}{\begin{equation}}
\newcommand{\eeq}{\end{equation}}
\newcommand{\bea}{\begin{array}}
\newcommand{\ena}{\end{array}}
\newcommand{\bds}{\begin {itemize}}
\newcommand{\eds}{\end {itemize}}
\newcommand{\bdf}{\begin{definition}}
\newcommand{\blm}{\begin{lemma}}
\newcommand{\edf}{\end{definition}}
\newcommand{\elm}{\end{lemma}}
\newcommand{\bthm}{\begin{theorem}}
\newcommand{\ethm}{\end{theorem}}
\newcommand{\bprp}{\begin{prop}}
\newcommand{\eprp}{\end{prop}}
\newcommand{\bcl}{\begin{claim}}
\newcommand{\ecl}{\end{claim}}
\newcommand{\bcr}{\begin{coro}}
\newcommand{\ecr}{\end{coro}}
\newcommand{\bquest}{\begin{question}}
\newcommand{\equest}{\end{question}}

\newcommand{\larrow}{{\larrow}}

\newcommand{\nin}{{\not \in}}

\def\urltilda{\kern -.15em\lower .7ex\hbox{\~{}}\kern .04em}